\documentclass[11pt]{article}
\usepackage{amsthm}
\usepackage{fullpage}
\usepackage[all,poly,arc]{xy}
\usepackage{color} 
\usepackage{soul} 

\newtheorem{thrm}{Theorem}[section]

\newtheorem{lemma}[thrm]{Lemma}
\newtheorem{corollary}[thrm]{Corollary}
\newtheorem{proposition}[thrm]{Proposition}

\newtheorem{definition}[thrm]{Definition}

\newtheorem{example}[thrm]{Example}
\newtheorem{fact}[thrm]{Fact}

\newcommand{\bw}{\mathbf{w}}
\newcommand{\emptyword}{\varepsilon}

\newenvironment{proof2}[1]{\par\noindent \textit{Proof of #1.} \rm}{} 

\title{On Quasiperiodic Morphisms}
\author{F. Lev\'e\footnote{Laboratoire MIS, 33 rue Saint Leu, 80039
    Amiens Cedex 1 - France}, 
G. Richomme\footnote {LIRMM (CNRS, Univ. Montpellier 2) - UMR 5506 -
  CC 477, 161 rue Ada, 34095, Montpellier Cedex 5 - France}
\footnote {Univ. Paul-Val\'ery Montpellier 3, Dpt MIAp, Route de Mende, 34199 Montpellier Cedex 5, France}
}

\date{}

\begin{document}
\maketitle

\begin{abstract}
Weakly and strongly quasiperiodic morphisms are tools introduced to study quasiperiodic words. 
Formally they map respectively at least one or any non-quasiperiodic word to a quasiperiodic word. Considering them both on finite and infinite words, we get four families of morphisms between which we study relations. We provide algorithms to decide whether a morphism is strongly quasiperiodic on finite words or on infinite words.
\end{abstract}

\section{Introduction}

The notion of quasiperiodicity we consider in this paper is the one introduced in the area of Text Algorithms by Apostolico and Ehrenfeucht \cite{ApostolicoEhrenfeucht1993TCS} in the following way: ``a
string $w$ is quasiperiodic if there is a second string $u \neq w$
such that every position of $w$ falls within some occurrence of $u$ in
$w$''. In 1994, Marcus extended this notion to right infinite words and he opened six questions. Four of them were answered in \cite{LeveRichomme2004BEATCS} (see also \cite{MarcusMonteil2004arxiv}). 
In particular, we proved the existence of a Sturmian word which is not quasiperiodic.

In \cite{LeveRichomme2007TCS}, we proved that a Sturmian word is not quasiperiodic if and only if it is an infinite Lyndon word.
The proof of this result was based on the $S$-adicity of Sturmian words
(Sturmian words form a family of non-periodic words 
that can be infinitely decomposed over four basic morphisms -- see \cite{BerstelSeebold2002Lothaire} for more properties on Sturmian words) and on a characterization of morphisms that preserve Lyndon words \cite{Richomme2003BBMS}. In \cite{LeveRichomme2007TCS}, we introduced strongly quasiperiodic morphisms as those morphisms that map all infinite words to quasiperiodic ones, and weakly quasiperiodic morphisms that map at least one non-quasiperiodic word to a quasiperiodic one. We characterized Sturmian morphisms that are strongly quasiperiodic and those that are not weakly quasiperiodic.

With Glen \cite{GlenLeveRichomme2008TCS}, previous results were extended to the class of episturmian words. All quasiperiodic episturmian words were characterized (unlike the Sturmian case, they do not correspond to infinite episturmian Lyndon words). Two proofs were provided for this result. The first one used connections between quasiperiodicity and return words, the second one used $S$-adic decompositions of episturmian words, and a characterization of strongly quasiperiodic on infinite words episturmian morphisms.

Observe that strongly and weakly quasiperiodic morphisms were considered in the context of infinite words. In this paper we consider also these morphisms with respect to finite words. 
After basic definitions (Sect.~\ref{sec:def}), in Sect.~\ref{sec:relations}, we study existing relations between the four so-defined families of morphisms. Algorithms to check if a morphism is strongly quasiperiodic are provided in Sect.~\ref{sec:decision_wqf} and \ref{sec:decision_wqi}. In Sect.~\ref{sec:weakly}, we provide sufficient conditions for a morphism to be weakly quasiperiodic on infinite words.

\section{Quasiperiodic Words and Morphisms\label{sec:def}}

We assume readers are familiar with combinatorics on words, morphisms and automata (see for instance \cite{Lothaire2002book}). 
We let $\varepsilon$ denote the empty word, $|w|$ denote the length of a word $w$, and $|w|_a$ denote the number of occurrences of a letter $a$ in $w$.
Let us recall that, if some words $w$, $u$, $p$ and $s$ verify $w =
ups$, then $p$ is called a \textit{prefix} of $w$, $s$ a
\textit{suffix} of $w$ and $u$ a \textit{factor} of $w$. A factor,
prefix or suffix is said to be \textit{proper} if it differs from the
whole word. An \textit{internal} factor of a word is any occurrence of
a factor except its prefixes and suffixes. For a word $u$ and an
integer $k$, $u^k$  denotes the word obtained by concatenating $k$
copies of $u$ and $u^\omega$ denotes the infinite periodic word obtained by concatenating infinitely many copies of $u$.

Given a non-empty word $q$, $q$-quasiperiodic words (or strings) are defined in the introduction. Equivalently a finite word $w$ is $q$-\textit{quasiperiodic} 
if $w \neq q$ and there exist words $p$, $s$ and $u$ such that 
$w = qu$, $q = ps$, $p \neq \varepsilon$, and $su = q$ or $su$ is a
$q$-quasiperiodic word. The word $q$ is called a \textit{quasiperiod}
of $w$. It is called \textit{the quasiperiod} of $w$ if $w$ has no
smaller quasiperiod. 
For instance, the word $w=ababaabababaabababa$ is $aba$-quasiperiodic and
$ababa$-quasiperiodic. The word $aba$ is 
the quasiperiod of $w$.

A word $w$ is said \textit{quasiperiodic} if it is $q$-quasiperiodic
for some word $q$. Otherwise $w$ is called \textit{superprimitive}. 
The quasiperiod of any quasiperiodic word $w$ is superprimitive.
The definition of quasiperiodicity extends naturally to infinite words.

Let us recall that a morphism $f$ is an application on words such that
for all words $u$ and $v$, $f(uv) = f(u)f(v)$. Such a morphism is
defined by images of letters. A well-known morphism is the Fibonacci
morphism $\varphi$ defined by $\varphi(a) = ab$, $\varphi(b)=a$. In~\cite{LeveRichomme2004BEATCS}, we proved that the infinite Fibonacci word, 
the fixed point of $\varphi$, has infinitely many quasiperiods that
are superprimitive. The first ones are $aba$, $abaab$, $abaababaa$.

Notice that from now on, we will only consider non-erasing morphisms (images of non-empty words differ from the empty word).
As mentioned in the introduction, \textit{strongly quasiperiodic on infinite words morphisms} were introduced as a tool to study quasiperiodicity of some infinite words. They are the morphisms that map any infinite word to a quasiperiodic infinite words. Also were introduced \textit{weakly quasiperiodic on finite words morphisms} that map at least one non-quasiperiodic infinite word to a quasiperiodic one. Examples are provided in the next section. It is interesting to observe that a morphism that is not weakly quasiperiodic on infinite words could be called a quasiperiodic-free morphism as it maps any non-quasiperiodic infinite word to another non-quasiperiodic word. This allows to relate the current study to the stream of works around power-free morphisms. In this context, it is natural to consider previous notions on finite words. Thus in this paper, we will also consider \textit{strongly quasiperiodic on finite words morphisms} that map any finite word to a quasiperiodic word, and \textit{weakly quasiperiodic on finite words morphisms} that map at least one finite non-quasiperiodic word to a quasiperiodic word.

\section{\label{sec:relations}Relations}

In this section, we show that the basic relations between the different families of morphisms are the ones described in Fig.~\ref{fig:basic_relations}.
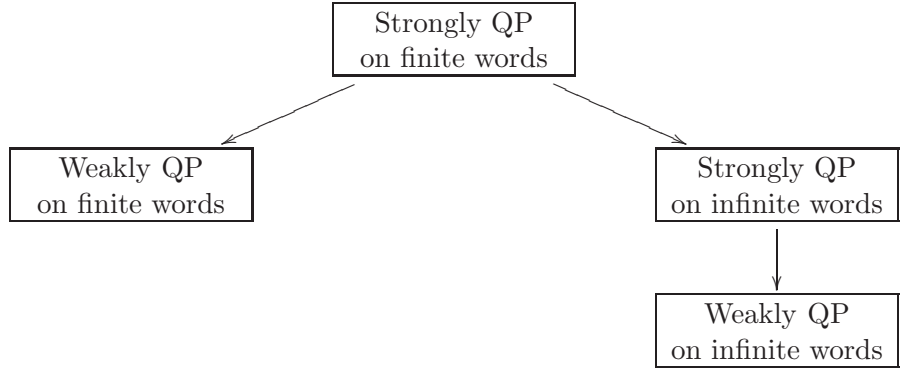
\begin{figure}
 $$\xymatrix{
 & \fbox{\begin{minipage}{3cm}\center Strongly QP\\on finite words\end{minipage}} \ar[dl]\ar[dr]& \\
\fbox{\begin{minipage}{3cm}\center Weakly QP\\on finite words\end{minipage}} & & \fbox{\begin{minipage}{3cm}\center Strongly QP\\on infinite words\end{minipage}}\ar[d]\\
&& \fbox{\begin{minipage}{3cm}\center Weakly QP\\on infinite words\end{minipage}}\\
}$$
\caption{Basic relations\label{fig:basic_relations}}
\end{figure}

Let us first observe that it follows the definitions that any strongly
quasiperiodic on finite (resp. infinite) words morphism is also a
weakly quasiperiodic on finite (resp. infinite) words
morphism. Next result proves the last relation of Fig.~\mbox{\ref{fig:basic_relations}}.
Its proof uses Lemma~\mbox{\ref{L:2}}.

\begin{proposition}
\label{P:1}
Any strongly quasiperiodic on finite words morphism is strong\-ly quasiperiodic on infinite words.
\end{proposition}

\begin{lemma}
\label{L:2}
Let $f$ be a morphism. 
Assume the existence of two words $u$ and $v$ and of an integer $k$ 
such that $|f(u)^k| \geq |f(v)|$. 
If $f(u)$ and $f(u^k v u^k)$ are quasiperiodic, then their quasiperiods are equal.
\end{lemma}

\begin{proof}
Let $q_u$ be the quasiperiod of $f(u)$ and let $q$ be the quasiperiod of the word $f(u^kvu^k)$.

If $|q|< |q_u|$, then $q$ is a prefix and a suffix of $q_u$ and as $f(u)$ is a factor of a $q$-quasiperiodic word, it is also $q$-quasiperiodic (we have $f(u) \neq q$ for length reason). This contradicts the fact that, by definition, $q_u$ is the smallest quasiperiod of $f(u)$.

So $|q_u| \leq |q|$.
Assume $|q| \geq 2|f(u^k)|$. So by choice of $k$, $|q| \geq |f(u^k)|+|f(v)|$.
This implies that the prefix occurrence of $q$ in $f(u^kvu^k)$ overlaps the suffix occurrence. More precisely $q = q_1 q_2 = q_2 q_3$ with $|q_1q_2| \geq 2|f(u^k)|$ and $|q_1| = |q_3| \leq |f(u^k)|$: we have $|q_2| \geq |q_1|$. By a classical result (see \cite[Lem. 1.3.4]{Lothaire1983book}), there exists words $x$ and $y$ with $xy \neq \emptyword$ and an integer $\ell$ such that $q_1 = xy$, $q_2 = (xy)^\ell x$ and $q_3 = yx$. For length reason, $\ell \neq 0$ so that $q$ is $xyx$-quasiperiodic. This contradicts the fact that $q$ is superprimitive.

Thus $|q| < 2|f(u^k)|$. As $q$ is both prefix and
suffix-comparable with $f(u^k)$ which is
$q_u$-quasiperiodic, as $|q_u| \leq |q|$,
and as $q$ is superprimitive, $q = q_u$. 
\end{proof}

\begin{proof2}{Proposition~\ref{P:1}}
Assume $f$ is strongly quasiperiodic on finite words. Let $\alpha$ be
a letter and let $q_\alpha$ be the quasiperiod of $f(\alpha)$. By Lemma~\ref{L:2}, for any word $u$, there exists an integer $k$ such that $f(\alpha^ku\alpha^k)$ is $q_\alpha$-quasiperiodic.
This implies that, for any word $u$, $f(\alpha u)$ is a prefix of a $q_\alpha$-quasiperiodic word. Equivalently, for any infinite word
$\bw$, $f(\alpha \bw)$ is a $q_\alpha$-quasiperiodic word.\qed
\end{proof2}

\medskip
Conversely to Proposition~\mbox{\ref{P:1}}, it is easy to find an example showing the existence of a morphism that is strongly quasiperiodic on infinite words but not on finite words. Just look at the morphism that maps $a$ to $aa$ and $b$ to $a$, or at next example of a strongly quasiperiodic morphism on
infinite words that is not weakly quasiperiodic on finite words.

\begin{example}\rm
Let $f$ be the morphism defined on $\{a, b\}^*$ by
\begin{quote}
$f(a) = abaababaababababaab$\\
$f(b) = abaabaabababababaab$.
\end{quote}
It is straigthforward that $f(\bw)$ is $aba$-quasiperiodic for any infinite word $\bw$.
Let us prove that $f$ is not weakly quasiperiodic on finite words.
Assume by contradiction the existence of a non-quasiperiodic word $u$
such that $f(u)$ is quasiperiodic. Observe $u \neq a$, $u \neq b$ and
the quasiperiod of $u$ ends with $ab$. An exhaustive verification
allows to see that no proper prefix of $f(a)$ nor $f(b)$ could be a
quasiperiod of $f(u)$. Hence $f(a)$ or $f(b)$ is a prefix of the
quasiperiod $q$ of $f(u)$. Observing this implies $|q| \geq
|f(a)| = |f(b)|$, we deduce that $f(a)$ or $f(b)$ is a suffix of $q$.
As $f(a)$ and $f(b)$ are not internal factors of $f(aa)$, $f(ab)$, $f(ba)$, $f(bb)$, $q = f(q')$ for some word $q'$. Moreover $u$ is $q'$-quasiperiodic, a contradiction.
\end{example}

\medskip
Next examples show that the other converses of the relations presented in Fig.~\ref{fig:basic_relations} are false.

\begin{example}
The morphism that maps $a$ to $aa$ and $b$ to $bb$ is weakly
quasiperiodic on finite words (as $f(a)$ is quasiperiodic), but we let
readers verify that it is not weakly quasiperiodic on infinite
words. Thus $f$ is not strongly quasiperiodic on infinite words and,
as a consequence of Proposition~\ref{P:1}, it is not strongly quasiperiodic on finite words.
\end{example}

\begin{example}
\label{ex:3.5}
The morphism $f$ defined by $f(a) = ba$ and $f(b) = bba$ is weakly quasiperiodic on infinite words since for all word $w \in a\{a,b\}^\omega$, $f(w)$ is $bab$-quasiperiodic. 
But $f(ba^\omega)=bb(ab)^\omega$ is not quasiperiodic, and so $f$ is not strongly quasiperiodic on infinite words. By Proposition~\ref{P:1}, {$f$ is not strongly quasiperiodic on finite words.}
\end{example}

\section{\label{sec:decision_wqf}Deciding Strong Quasiperiodicity on Finite Words}

Next lemma which is a direct consequence of Lemma~\ref{L:2} is the key observation to decide whether a morphism is strongly quasiperiodic on finite words.

\begin{lemma}
\label{L:borne_sqf}
If $f$ is a strongly quasiperiodic on finite words morphism, then for any word $u$ and any letter $\alpha$, the quasiperiod of $f(u)$ is a factor of $f(\alpha^3)$ of length less than $2|f(\alpha)|$.
\end{lemma}

\begin{proof}
Assume $f$ is strongly quasiperiodic on finite words.
Let $u$ be a word and let $q_u$ be the quasiperiod of $f(u)$.
Let $i$ be an integer such that $|f(\alpha^i)| \geq 2|q_u|$ ($|f(\alpha)| \neq 0$ as $f(\alpha)$ is quasiperiodic).
Let $k$ be an integer such that $|f(u^k)| \geq |f(\alpha^i)|$.
By Lemma~\ref{L:2}, the quasiperiod of $f(u^k\alpha^i u^k)$ is $q_u$. As $|f(\alpha)^i| \geq 2|q_u|$, $q_u$ must be a factor of $f(\alpha)^i$.
As $q_u$ is superprimitive, $|q_u| < 2|f(\alpha)|$.
Consequently $q_u$ is a factor of $f(\alpha)^3$. 
\end{proof}

Observe now that, given two words $u$ and $q$, it follows the
definition of quasiperiodicity that the $q$-quasiperiodicity of $f(u)$
implies that, for each non-empty proper prefix $\pi$ of $f(u)$, $\pi =
xps$ with $xp = \varepsilon$, $xp = q$ or $xp$ is the longest
$q$-quasiperiodic prefix of $\pi$ if $|\pi|>|q|$, and
$ps$ a prefix of $q$. 
{Based on this remark, we introduce an automaton that will allow to recognize words $u$ such that $f(u)$ is $q$-quasiperiodic (or $q$ or the empty word $\varepsilon$), for a given word $q$ and a given morphism $f$.
}
{Note that a quasiperiod may have several borders, that is, proper suffixes that are prefixes.}
{For instance, the word $q = abacaba$ has $\varepsilon$, $a$ and $aba$ as borders.}
{Thus while processing the automaton, one cannot determine with precision which will be the word $p$ occurring in previous observation until the reading of next letters.}
{Therefore the constructed automaton will just remind (instead of initial $p$) the longest suffix $p$ of $\pi$ such that $ps$ is a prefix of $q$.}

\begin{definition}
\label{D:automate1}
Let $f$ be a morphism over $A^*$ and $q$ be a non-empty word. We denote ${\cal A}_q(f)$, or simply ${\cal A}_q$, the automaton $(A, Q, i, F, \Delta)$ where:
\begin{itemize}
\item the states, the elements of $Q$, are the couples $(p, s)$ such that $ps$ is a proper prefix of $q$;
\item the initial state $i$ is the couple $(\varepsilon, \varepsilon)$;
\item the final states, the elements of $F$, are the couples on the form $(p, \varepsilon)$, with $p$ a prefix of $q$;
\item the transitions, the elements of $\Delta$, are triples $((p_1, s1), a, (p_2, s_2))$ where $(p_1, s1) \in Q$, $(p_2, s_2) \in Q$ and one of the two following situations holds:
 \begin{enumerate}
 \item 
If $q$ does not occur in $p_1s_1f(a)$ and $|q| > |s_1f(a)|$, then
\begin{itemize}
\item
$s_1f(a)=s_2$, 
\item $p_2$ is the longest suffix of $p_1$ such that $p_2s_1f(a)$ is a
  proper prefix of $q$.
\end{itemize}

 \item If $q$ occurs in $p_1s_1f(a)$
\begin{itemize}
\item
 there exist a suffix $x$ of $p_1$ and a word $y$ such that $xs_1f(a)
 = ys_2$ with $y = q$ or $y$ is $q$-quasiperiodic,
\item
$p_2$ is the longest suffix of $y$ such that $p_2s_2$ is a proper prefix of
$q$. 
\end{itemize}
 \end{enumerate}
\end{itemize}
\end{definition}

{The automaton} defined in previous definition {is determinist}. It should be
emphasized that given a state $(p, s)$ and a letter $a$, there may not
exist a state $(p', s')$ such that a transition $((p,s), a, (p',s'))$
exists. {We let readers verify the next observation and its corollary.} 

\begin{fact}
\label{F:fonctionnement_automate}
Any state $(p,s)$ in ${\cal A}_q$ is reached by reading a word $u$ if
and only if there exist words $\pi$, $p$ and $s$, such that $f(u) =
\pi p s$ with $\pi p = \varepsilon$, $\pi p = q$ or $\pi p$ is a $q$-quasiperiodic word, and, $ps$ is the longest prefix of $q$ that is a suffix of $f(u)$.
\end{fact}

\begin{lemma}
\label{L:automate}
A word $u$ is recognized by ${\cal A}_q$ if and only if $f(u) = \varepsilon$ or $f(u) = q$ or $f(u)$ is $q$-quasiperiodic.
\end{lemma}

Let us give some examples of automata following the previous
definition. Notice that we just construct the states that are accessible from $(\varepsilon, \varepsilon)$.

\begin{example}
\label{ex:aut1}
Let $f$ be the morphism defined by $f(a) = ab$, $f(b) = aba$. The automaton ${\cal A}_{aba}$ is the following one.

 $$
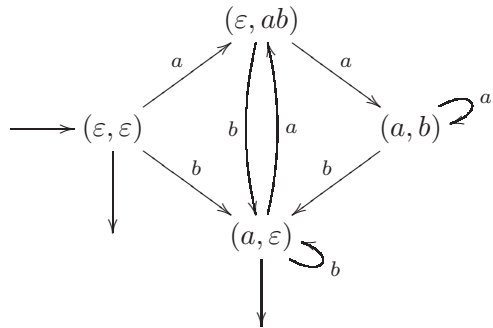
\xymatrix{
& & (\varepsilon, ab)\ar[dr]^a\ar@/_/[dd]_b & \\
\ar[r]&\ar[d](\varepsilon, \varepsilon)\ar[ur]^a\ar[dr]^b && (a, b)\ar[dl]_b\ar@(ru,r)[]^a\\
&& (a, \varepsilon)\ar@/_/[uu]_a\ar[d]\ar@(rd,r)[]_b \\
&&&\\
}$$
\end{example}

\begin{example}
\label{ex:aut2}
Let $f$ be the morphism defined by $f(a) = abaaba$, $f(b) = baabaaba$. 
Here follow automata ${\cal A}_{aba}$ and ${\cal A}_{baaba}$.

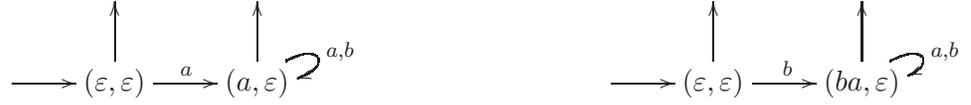
\hfill\begin{minipage}{5cm}
 $$\xymatrix{
&&&\\ 
\ar[r]&\ar[u](\varepsilon, \varepsilon)\ar[r]^a & (a, \varepsilon)\ar[u]\ar@(ru,r)[]^{a,b}\\
}$$
\end{minipage}\hfill
\begin{minipage}{5cm}
 $$\xymatrix{
&&&\\ 
\ar[r]&\ar[u](\varepsilon, \varepsilon)\ar[r]^b & (ba, \varepsilon)\ar[u]\ar@(ru,r)[]^{a,b}\\
}$$
\end{minipage}\hfill
\end{example}

\begin{example}
\label{ex:aut3}
Let $f$ be the morphism defined by $f(a) = aabaab$, $f(b) = aabaaaba$ and $f(c) = aabaababaabaa$. 
Here follows automaton ${\cal A}_{aabaa}$.

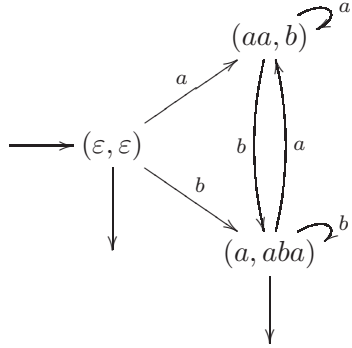
 $$\xymatrix{
& & (aa, b)\ar@/_/[dd]_b\ar@(ru,r)[]^a  \\
\ar[r]&\ar[d](\varepsilon, \varepsilon)\ar[ur]^a\ar[dr]^b \\
&& (a, aba)\ar@/_/[uu]_a\ar[d]\ar@(ru,r)[]^b \\
&&&\\
}$$
\end{example}

Let ${\cal Q}(f)$ be the set of all words $q$ such that, for all letters $\alpha$ in $A$, $|q| \leq 2|f(\alpha)|$ and $q$ is a factor of $f(\alpha)^3$.
Following  Lemma~\ref{L:borne_sqf}, ${\cal Q}(f)$ is the set of all possible quasiperiods of a word on the form $f(u)$. Thus Lemma~\ref{L:automate} implies the next characterization of strongly quasiperiodic morphisms on finite words.

\begin{proposition}
\label{P:strong_qp_on_finite}
A morphism $f$ is strongly quasiperiodic on finite words if and only if, for each letter $\alpha$, the word $f(\alpha)$ is quasiperiodic, and 
\[
A^* = \bigcup_{q \in {\cal Q}(f)} {\cal L}({\cal A}_q)
\]
where ${\cal L}({\cal A}_q)$ is the language recognized by the automaton ${\cal A}_q$.
\end{proposition}

As ${\cal Q}(f)$ is finite, and as it is decidable whether a finite word is quasiperiodic~\cite{ApostolicoEhrenfeucht1993TCS,BrodalPedersen2000,IliopoulosMouchard1999a} (see also \cite{GroultRichomme2010TCS} for optimality of the complexity of these algorithms), we can conclude.

\begin{corollary}
It is decidable whether a morphism is strongly quasiperiodic on finite words.
\end{corollary}

To end this section, let us illustrate Proposition~\ref{P:strong_qp_on_finite}.
If $f$ is the morphism considered in Example~\ref{ex:aut2} ($f(a) = abaaba$, $f(b) = baabaaba$), as $aba$ and $baaba$ belong to ${\cal Q}(f)$, as ${\cal L}({\cal A}_{aba}) = \varepsilon \cup a \{a, b\}^*$ and 
${\cal L}({\cal A}_{baaba}) = \varepsilon \cup b \{a, b\}^*$, as $f(a)$ and $f(b)$ are quasiperiodic, we can conclude by Proposition~\ref{P:strong_qp_on_finite} that $f$ is strongly quasiperiodic on finite words.

Now consider the morphism defined by $f(a) = ab$, $f(b) = aba$. We have ${\cal Q}(f) = \{a, b, ab, ba, aba\}$.
By Example~\ref{ex:aut1}, ${\cal L}({\cal A}_{aba}) = \varepsilon \cup \{a, b\}^*b$. We let readers verify that ${\cal L}({\cal A}_{a}) = {\cal L}({\cal A}_{b}) = {\cal L}({\cal A}_{ba}) = \emptyset$ and ${\cal L}({\cal A}_{ab}) = a^*$. Thus $f$ is not strongly quasiperiodic. As the set ${\cal L}({\cal A}_{aba})$ contains non-quasiperiodic words, this morphism $f$ is weakly quasiperiodic.

\section{\label{sec:decision_wqi}Deciding Strong Quasiperiodicity on Infinite Words}

We now show how to adapt the ideas of previous section to the study of strongly quasiperiodic on infinite words morphisms. First we adapt Lemma~\ref{L:borne_sqf}.

\begin{lemma}
\label{L:borne_sqi}
If $f$ is a strongly quasiperiodic on infinite words morphism, then for any infinite word $\bw$ and any letter $\alpha$, the quasiperiod of $f(\bw)$ is a factor of $f(\alpha^3)$ of length less than $2|f(\alpha)|$ that is a factor of ${\cal Q}(f)$.
\end{lemma}

This result is a consequence of the next one whose proof is similar to the one of Lemma~\ref{L:borne_sqf} (without the need of Lemma~\ref{L:2}).

\begin{lemma}
\label{L:borne_sqi_etape}
If $f$ is a strongly quasiperiodic on infinite words morphism, then for any word $u$ and any letter $\alpha$, the quasiperiod of $f(u\alpha^\omega)$ is a factor of $f(\alpha^3)$ of length less than $2|f(\alpha)|$.
\end{lemma}

\begin{proof2}{Lemma~\ref{L:borne_sqi}}
Let $f$ be a strongly quasiperiodic on infinite words morphism.
Let $\bw$ be an infinite word and let $\alpha$ be a letter.
With each prefix $p$ of $\bw$, by Lemma~\ref{L:borne_sqi_etape}, 
one can associate a factor $q_p$ of $f(\alpha^3)$ such that $f(p\alpha^\omega)$ is $q_p$-quasiperiodic. As the set of factors of $f(\alpha^3)$ is finite, there exists one, say $q$, which is associated with an infinity of prefixes of $\bw$. This implies $\bw$ is $q$-quasiperiodic. \qed
\end{proof2}

Now we adapt the automaton used in the previous section in order to have a tool to determine if 
the image of an infinite word is $q$-quasiperiodic for a given morphism and a given word $q$.

\begin{definition}
Let $f$ be a morphism over $A^*$ and $q$ be a non-empty word. Let ${\cal A}_q'(f)$, or simply ${\cal A}_q'$, denote the automaton $(A, Q, i, F', \Delta)$ where $Q$, $i$, $\Delta$ are defined as in Definition~\ref{D:automate1}, and $F' = Q$.
\end{definition}

\begin{lemma}
\label{L:automate2}
An infinite word $\bw$ is $q$-quasiperiodic if and only if all its prefixes are recognized by ${\cal A}_q'$.
\end{lemma}

As a consequence of Lemmas~\ref{L:borne_sqi} and \ref{L:automate2}, we get 
next characterization of strongly quasiperiodic morphisms on finite words.

\begin{proposition}
\label{P:carac_strong_qp_on_infinite}
A morphism $f$ is strongly quasiperiodic on infinite words if and only if 
\[
A^* = \bigcup_{q \in {\cal Q}(f)} {\cal L}({\cal A}_q')
\]
where ${\cal L}({\cal A}_q')$ is the language recognized by the automaton ${\cal A}_q'$.
\end{proposition}

The proof of Proposition~\ref{P:carac_strong_qp_on_infinite} is a consequence of the previous definition and lemmas. To make all clearer, just observe that, if a word $u$ is recognized by ${\cal A}_q'$ then all its prefixes are also recognized.

As an example to illustrate Proposition~\ref{P:carac_strong_qp_on_infinite}, one can consider the morphism $f$ defined by $f(a) = ab$, $f(b) = aba$. Example~\ref{ex:aut1} shows that ${\cal A}_{aba}' = \{a, b\}^*$ and so $f$ is strongly quasiperiodic on infinite words.

On the same way, one can verify that the morphism $f$ defined by $f(a) = abaaba$ and $f(b) = aabaaba$ is strongly-quasiperiodic.
More precisely, the image of any infinite word beginning with $a$ is $abaa$-quasiperiodic and the image of any word beginning with $b$ is $aaba$-quasiperiodic. 

As a consequence of Proposition~\ref{P:carac_strong_qp_on_infinite}, we have next result.

\begin{corollary}
It is decidable whether a morphism is strongly quasiperiodic on infinite words.
\end{corollary}

\section{\label{sec:weakly}On Weakly Quasiperiodic Morphisms}

We now consider the decidability of the questions: given a morphism $f$, is $f$ weakly quasiperiodic on finite words? Is it weakly quasiperiodic on infinite words? Note that this is equivalent to asking for the decidability of the question: given a morphism, are all images of non-quasiperiodic words also non-quasiperiodic? We provide some partial answers.

Let us recall that a morphism $f$ is said \textit{prefix}
(resp. \textit{suffix}) if for all letters $a$ and $b$, $f(a)$ is not
a prefix (resp. a suffix) of $f(b)$.

\begin{lemma}
\label{L:prefixsuffix}
Any non-prefix or non-suffix non-erasing morphism defined on an alphabet of cardinality at least two is weakly quasiperiodic on finite and infinite words.
\end{lemma}

\begin{proof}If $f(a)$ is a prefix of $f(b)$ then, for all $k \geq 1$, the
  finite word $f(b^ka)$ is $f(ba)$-quasiperiodic. The infinite word
  $f(bab^\omega)$ is also $f(ba)$-quasiperiodic. The morphism $f$ is
  weakly quasiperiodic both on finite words and on infinite words.

If $f(a)$ is a suffix of $f(b)$ then, for all $k \geq 1$, the
  finite word $f(ab^k)$ is $f(ab)$-quasiperiodic. The infinite word
  $f(ab^\omega)$ is  $f(ab)$-quasiperiodic (it is even periodic). The morphism $f$
  is weakly quasiperiodic both on finite words and on infinite words.
\end{proof}

\begin{corollary}
Any non-injective non-erasing morphism defined on an alphabet of cardinality at least two is weakly quasiperiodic on finite and infinite words.
\end{corollary}

\begin{proof}
If $f$ is not injective, there exist two different words $u$ and $v$ such that $f(u) = f(v)$. If $f(u)$ and $f(v)$ are powers of  same word then $f$ is erasing: a contradiction.
Otherwise, we can assume that $u$ and $v$ begin with different letters. Thus $f$ is not prefix and so, by Lemma~\ref{L:prefixsuffix}, it is weakly quasiperiodic on finite and infinite words.
\end{proof}

\begin{proposition}
\label{P:primitive}
Let $f$ be a non-erasing morphism and let $u$ be a primitive word over $\{a, b\}$. 
If $f(u)$ is not primitive then $f$ is weakly quasiperiodic on finite words.
Moreover, if $|u|_a \geq 1$ and $|u|_b \geq 1$, then $f$ is weakly quasiperiodic on infinite words.
\end{proposition}

We first need an intermediate result.

\begin{lemma}
\label{L:primitive}
If $f(a^ib^j)$ is not primitive for some integers $i \geq 1$, $j \geq 1$, then one of the words $f(ab^\omega)$, $f(aba^\omega)$, $f(ba^\omega)$, $f(bab^\omega)$ is quasiperiodic.
\end{lemma}

\begin{proof}
Assume first $i \geq 2$, $j \geq 2$. By Lyndon-Sch\"utzenberger's characterization of solutions of the equation $x^i y^j = z^k$ when $i \geq 2$, $j \geq 2$, $k \geq 2$ \cite{LyndonSchutz}, we deduce that $f(a)$ and $f(b)$ are powers of a same word: $f(ab^\omega)$ is quasiperiodic, as any image of a finite (of length at least 2) or of an infinite word.

Now consider case $j = 1$. Let $u$ be the primitive word such thay $f(a^ib) = u^k$ ($k \geq 2$). If $|f(a)^{i-1}| \geq |u|$, the words $f(a)^i$ and $u^k$ share a common prefix of length at least $|f(a)| + |u|$. By Fine and Wilf's theorem \cite{FineWilf1965}, $f(a)$ and $u$ are powers of a same word. It follows that $f(a)$ and $f(b)$ are also powers of a same word. We conclude as in case $i, j \geq 2$.

Now consider the case $|u| \geq |f(a)^i|$. From $f(a)^if(b) = u^k$, we get $u = f(a)^ix$, $f(b) = xu^{k-1}$ for some word $x$. Hence $f(b) = x(f(a)^ix)^{k-1}$ and the word $f(bab^\omega)$ is $x(f(a)^ix)$-quasiperiodic.

It remains to consider the case  $|f(a)^{i-1}| < |u| < |f(a)^i|$. In this case, for some words $x$ and $y$, $u = f(a)^{i-1}x$, $f(a) = xy$ and $y$ is a prefix of $u$. In particular, for some word $z$, $f(a) = xy = yz$. By  a classical result in Combinatorics on Words (see \cite[Lem. 1.3.4]{Lothaire1983book}), $x = \alpha\beta$, $y = (\alpha \beta)^\ell \alpha$, $z = \beta\alpha$: $f(a) = (\alpha \beta)^{\ell+1}\alpha$, $u = [(\alpha\beta)^{\ell+1}\alpha]^{i-1}\alpha\beta$. Now observe that 
$yf(b) = u^{k-1} = [[(\alpha\beta)^{\ell+1}\alpha]^{i-1}\alpha\beta]^{k-1}$.
When $ i\geq 2$, $f(b) = \beta\alpha [(\alpha\beta)^{\ell+1}\alpha]^{i-2}\alpha\beta[[(\alpha\beta)^{\ell+1}\alpha]^{i-1}\alpha\beta]^{k-2}$, and when $i = 1$,
$f(b) = \beta(\alpha\beta)^{k-\ell-2}$. In both cases, $f(aba^\omega)$ is $\alpha\beta\alpha$-quasiperiodic.

When $i = 1$, the non-primitivity of $f(ab^j)$ is equivalent to the non-primitivity of $f(b^ja)$. Thus exchanging the roles of $a$ and $b$, we end the proof of the lemma. 
\end{proof}

\begin{proof2}{Proposition~\ref{P:primitive}}
First if $u$ contains only the letter $a$ or only the letter $b$, we have $u = a$ or $u =b$ and $f$ is weakly quasiperiodic on finite words. Assume from now on that $|u|_a \geq 1$ and $|u|_b \geq 1$.
If $|u|_a = 1$, then there exist integers $i, j$ such that $u =
b^iab^j$ with $i+j \geq 1$. As $f(u)$ is not primitive, also
$f(ab^{i+j})$ is not primitive: $f$ is weakly quasiperiodic on
  finite words. By Lemma~\ref{L:primitive}, $f$ is also weakly
quasiperiodic on infinite words.
The result follows similarly when $|u|_b = 1$.
Now consider the case $|u|_a \geq 2$ and $|u|_b \geq 2$. A seminal result by Lentin and Sch\"utzenberger states that if $f$ is a morphism defined on alphabet $\{a, b\}$, if for a non-empty word $u$, $f(u)$ is not primitive then there exists a word $v$ in $a^*b \cap ab^*$ such that $f(v)$ is not primitive \cite[Th. 5]{LentinSchutzenberger1969}. We are back to previous cases.\qed
\end{proof2}

\medskip

The converse of Proposition~\mbox{\ref{P:primitive}} is false. Indeed
as shown by the morphism  $f$ defined by $f(a) = ababa$, $f(b) = ab$, a morphism
can be weakly quasiperiodic on finite words or on infinite words and
be primitive preserving (the image of any primitive word is
primitive).
Nevertheless observe that when we consider the problem of deciding
if a morphism is weakly quasiperiodic on infinite words, we can assume
that all images of letters are primitive. Indeed any morphism $f$ such that $f(a)$ is a non-empty power of $a$ for each letter $a$ is not weakly quasiperiodic: for any word (finite of length at least 2 or infinite) $w$, $f(w)$ is quasiperiodic if and only if $w$ is quasiperiodic. In consequence, to determine whether a morphism $f$ is weakly quasiperiodic or not, one can substitute $f$  by the morphism $r_f$ defined by $r(a)$ is the primitive root of $f(a)$. Note that images of letters by $r_f$ are primitive words.

\medskip
For all weakly quasiperiodic on infinite words morphisms met until now, there exist non-empty words $u$ and $v$ such that the infinite word $uv^\omega$ is not quasiperiodic while $f(uv^\omega)$ is quasiperiodic.
This situation also holds in the next lemma (when $\bw$ in the hypothesis is not quasiperiodic) whose proof is omitted. We conjecture that this holds in all cases. Bounding the length of $u$ and $v$ could lead to a procedure to check whether a morphism is weakly quasiperiodic on infinite words.

\begin{lemma}
\label{L:demiqp}
Let $f$ be a morphism, and let $\bw$ be an infinite word such that $f(\bw)$ is $q$-quasiperiodic for some word $q$ such that $2|q| \leq |f(\alpha)|$ for each letter $\alpha$. Then:
\begin{enumerate}
\item $\bw = (a_1\ldots a_k)^\omega$ with $a_1$, \ldots, $a_k$ pairwise different letters, or,
\item there exist words $x$, $y$, $z$ and letters $a$ and $b$ such
  that $|xyz|_a = 0$, $|z|_b = 0$, $xay(bz)^\omega$ is not
  quasiperiodic and $f(xay(bz)^\omega)$ is $q$-quasiperiodic. Moreover
  in this case, we can find $x$, $y$ and $z$ such that any letter
  occurs at most once {in} each of these words.
\end{enumerate}
\end{lemma}

\section{\label{sec:conclusion}Conclusion}

To conclude this paper on links between quasiperiodicity and
morphisms, we point out another question. Given a morphism $f$
prolongable on a letter $a$, can we decide whether the word
$f^\omega(a) = \lim_{n \to \infty} f^n(a)$ is quasiperiodic? We are
convinced that a better knowledge of weakly and strongly quasiperiodic
on infinite words {morphisms} could bring answers to the previous
question. We suspect in particular that if $f$ is a strongly
quasiperiodic on infinite words morphism and if it is prolongable on
$a$, then $f^\omega(a)$ is quasiperiodic. Conversely it should be true
that if $f^\omega(a)$ is quasiperiodic and $f(a)$ is not a power of
$a$ then $f$ is weakly quasiperiodic on infinite words. The next result states partially that.

\begin{proposition}
\label{P:fixedpoint}
Let $f$ be a non-erasing morphism and $a$ be a letter such that $f^\omega(a)$ is a quasiperiodic infinite word but not a periodic word. 
If all letters are growing with respect to $f$ ($\lim_{n \to \infty} |f^n(a)|$ $= \infty$), 
then $f$ is weakly quasiperiodic on infinite words.
\end{proposition}

Observe that the converse of previous proposition does not hold. The
morphism $f$ defined by $f(a) = a$, $f(b) = ba$ does not generate an
infinite quasiperiodic word ($f$ does not generate its fixed point
  $a^\omega$ and $ba^\omega$ is not quasiperiodic), but it is weakly quasiperiodic on infinite words as $f(ab^\omega)$ is $aba$-quasiperiodic.

It is an open {problem} to state Proposition~\ref{P:fixedpoint} for
arbitrary morphims generating a quasiperiodic {infinite word}.

The proof of Proposition~\ref{P:fixedpoint} is a consequence of {Lemma}~\ref{L:demiqp} and the following one.

\begin{lemma}
\label{L:reducing}
Let $f$ be a non-erasing morphism.
If, for some integer $k \geq 1$, the morphism $f^k$ is weakly quasiperiodic, then $f$ is weakly quasiperiodic.
\end{lemma}

\begin{proof}
Assume $f^k(\bw)$ is quasiperiodic for some integer $k \geq 1$ and for some non-quasi\-pe\-riodic infinite word $\bw$.
Let $i$ be the smallest integer such that $f^i(w)$ is quasiperiodic. Observe that $i \geq 1$ and that $f^{i-1}(w)$ is not quasiperiodic. As $f^i(\bw) = f(f^{i-1}(\bw))$, $f$ is weakly quasiperiodic on infinite words. 
\end{proof}

\begin{proof2}{Proposition~\ref{P:fixedpoint}}
Let $f$ be a morphism and let $a$ be a letter such that $f^\omega(a)$ is a
quasiperiodic infinite word. Let $q$ be the quasiperiod of $f^\omega(a)$. Assume
that all letters of $f$ are growing. 
As all letters are growing with respect to $f$, for some $k \geq 1$, $f^k$ verifies the hypothesis of Lemma~\mbox{\ref{L:demiqp}}: $f^k$ is weakly
quasiperiodic on infinite words. 
By Lemma~\ref{L:reducing}, $f$ is also weakly quasiperiodic on infinite words. \qed
\end{proof2}

\bibliographystyle{plain}
\bibliography{local}

\begin{thebibliography}{10}

\bibitem{ApostolicoEhrenfeucht1993TCS}
A.~Apostolico and A.~Ehrenfeucht.
\newblock Efficient detection of quasiperiodicities in strings.
\newblock {\em Theoret. Comput. Sci.}, 119:247--265, 1993.

\bibitem{BerstelSeebold2002Lothaire}
J.~Berstel and P.~S\'e\'ebold.
\newblock Sturmian words.
\newblock In M.~Lothaire, editor, {\em Algebraic Combinatorics on Words},
  volume~90 of {\em Encyclopedia of Mathematics and its Applications}, pages
  45--110. Cambridge University Press, 2002.

\bibitem{BrodalPedersen2000}
G.~S. Brodal and C.~N.~S. Pedersen.
\newblock Finding maximal quasiperiodicities in strings.
\newblock In {\em Combinatorial Pattern Matching (CPM'2000), 11th Annual
  Symposium, CPM 2000, Montreal, Canada, June 21-23, 2000}, volume 1848 of {\em
  Lecture Notes in Comput. Sci.}, pages 397--411, 2000.

\bibitem{FineWilf1965}
N.~J. Fine and H.~S. Wilf.
\newblock Uniqueness theorems for periodic functions.
\newblock {\em Proc. Amer. Math. Soc.}, 16:109--114, 1965.

\bibitem{GlenLeveRichomme2008TCS}
A.~Glen, F.~Lev\'e, and G.~Richomme.
\newblock Quasiperiodic and {L}yndon episturmian words.
\newblock {\em Theoret. Comput. Sci.}, 409(3):578--600, 2008.

\bibitem{GroultRichomme2010TCS}
R.~Groult and G.~Richomme.
\newblock Optimality of some algorithms to detect quasiperiodicities.
\newblock {\em Theoret. Comput. Sci.}, 411:3110--3122, 2010.

\bibitem{IliopoulosMouchard1999a}
C.~S. Iliopoulos and L.~Mouchard.
\newblock Quasiperiodicity: from detection to normal forms.
\newblock {\em Journal of Automata, Languages and Combinatorics},
  4(3):213--228, 1999.

\bibitem{LentinSchutzenberger1969}
A.~Lentin and M.~P. Sch{\"{u}}tzenberger.
\newblock A combinatorial problem in the theory of free monoids.
\newblock In R.C. Bose and T.W. Dowling, editors, {\em Combinatorial
  Mathematics and its Applications}, pages 128--144. Univ. North Carolina
  Press, 1969.

\bibitem{LeveRichomme2004BEATCS}
F.~Lev\'e and G.~Richomme.
\newblock Quasiperiodic infinite words: some answers.
\newblock {\em Bull. Eur. Assoc. Theor. Comput. Sci. EATCS}, 84:128--238, 2004.

\bibitem{LeveRichomme2007TCS}
F.~Lev{\'e} and G.~Richomme.
\newblock Quasiperiodic {S}turmian words and morphisms.
\newblock {\em Theoret. Comput. Sci.}, 372(1):15--25, 2007.

\bibitem{Lothaire1983book}
M.~Lothaire.
\newblock {\em Combinatorics on Words}, volume~17 of {\em Encyclopedia of
  Mathematics and its Applications}.
\newblock Addison-Wesley, 1983.
\newblock Reprinted in the {\em Cambridge Mathematical Library}, Cambridge
  University Press, UK, 1997.

\bibitem{Lothaire2002book}
M.~Lothaire.
\newblock {\em Algebraic Combinatorics on Words}, volume~90 of {\em
  Encyclopedia of Mathematics and its Applications}.
\newblock Cambridge University Press, 2002.

\bibitem{LyndonSchutz}
R.~C. Lyndon and M.-P. Sch{\"{u}}tzenberger.
\newblock The equation $a^m = b^n c^p$ in a free group.
\newblock {\em Michigan Math. J.}, 9:289--298, 1962.

\bibitem{MarcusMonteil2004arxiv}
S.~Marcus and T.~Monteil.
\newblock Quasiperiodic infinite words : multi-scale case and dynamical
  properties.
\newblock Technical Report arXiv:math/0603354, arxiv.org, 2006.

\bibitem{Richomme2003BBMS}
G.~Richomme.
\newblock Lyndon morphisms.
\newblock {\em Bull. Belg. Math. Soc. Simon Stevin}, 10(5):761--786, 2003.

\end{thebibliography}

\newpage

\section*{Appendix}

\begin{proof2}{Lemma~\ref{L:demiqp}}
Let $a$ be the first letter of $\bw$. Immediate consequences of the hypotheses ``{$2|q|\leq |f(\alpha)|$} for all letters {$\alpha$}'',  and ``$f(\bw)$ $q$-quasiperiodic'' are:
\begin{enumerate}
\item for any factor of $\bw$ on the form $au\alpha$ {with $u$ a word and $\alpha$ a letter}, there exists a prefix $p$ of $f(\alpha)$ such that $q$ is a suffix of $p$ and $f(au)p$ is $q$-quasiperiodic.
\item for any factor of $\bw$ on the form $\alpha u \beta$ {with
  $u$ a word and $\alpha, \beta$ letters}, there exist a suffix $s$ of $f(\alpha)$ and a prefix $p$ of $f(\beta)$ such that $q$ is a prefix of $s$ and a suffix of $p$, and such that $sf(u)p$ is $q$-quasiperiodic.
\item for any letter $\alpha$ occurring in $\bw$, if $f(\alpha) = xyz$ with {$x, y, z$ words and} $q$ both a prefix and a suffix of $y$, then $y$ is $q$-quasiperiodic (or $y = q$).
\end{enumerate}
It follows that, if $(a_i)_{i \geq 1}$ is a sequence of letters and $(u_i)_{i \geq 1}$ is a sequence of words such that $a_1 = a$ and for all $i \geq 1$, $a_iu_ia_{i+1}$ is a factor of $\bw$, the word $f(\prod_{i \geq 1} a_i u_i)$ is $q$-quasiperiodic. {In particular, if $au\alpha$ and $\alpha v\alpha$ are particular factors of $\bw$ then $f(au(\alpha v)^\omega)$ is $q$-quasiperiodic, or if $au\alpha$, $\alpha v\beta$ and $\beta w \beta$ are particular factors of $\bw$ then $f(au\alpha v(\beta w)^\omega)$ is $q$-quasiperiodic. For the same reason, if $f(x\alpha y \alpha \bw'$) is $q$-quasiperiodic with $\alpha$ a letter then $f(x\alpha \bw')$ is also $q$-quasiperiodic.}

If the letter $a$ is not recurrent in $\bw$, $\bw$ can be decomposed $\bw = axaubvb\bw'$ or $\bw = aubvb\bw'$ with $a$ that does not occur in $ubvb\bw'$. The word $au(bv)^\omega$ is not quasiperiodic while $f(au(bv)^\omega)$ is $q$-quasiperiodic.

Assume now that the letter $a$ is recurrent. If there exists a letter $b$ that is not recurrent, one can find two factors $axa$ and $aybza$ with $|x|_a = 0$, $|x|_b = 0$, $|ybz|_a = 0$. The word $aybz(ax)^\omega$ is not quasiperiodic while $f(aybz(ax)^\omega)$ is $q$-quasiperiodic.

Now assume that all letters of $\bw$ {are recurrent} and assume that
$\bw$ is not on the form $(a_1\ldots a_k)^\omega$ with $a_1$, \ldots,
$a_k$ pairwise different letters. There must exist two letters $b$ and
$c$, and a word {$v$} such that $bvb$ is a factor of $\bw$ and $|bvb|_c
= 0$. By recurrence, there exist $x$ and $y$ such that $\bw$ has
$bxcyb$ as a factor and $|xcy|_b = 0$. Moreover there
exists a factor $azb$ in $\bw$. The word $azbxcy(bv)^\omega$ is not
quasiperiodic while its image by $f$ is $q$-quasiperiodic.

Now, we observe that in all cases, when {$\bw \neq (a_1\ldots
a_k)^\omega$}{, we} have found words $w_1$, $w_2${,} $w_3$ and letters $\alpha${, } $\beta$ such that $|w_3|_\alpha = 0$, $|w_3|_\beta = 0$, {$w_1\alpha w_2(\beta w_3)^\omega$} is not quasiperiodic and {$f(w_1\alpha w_2(\beta w_3)^\omega)$} is $q$-quasiperiodic.

Observe that if $\alpha$ occurs in $w_1$, say $w_1 = w_4\alpha w_5$, then we can
replace $w_1$ by {$w_4$} with the same result. Thus we can assume
$|w_1|_\alpha = 0$. Similarly, we can assume that $|w_2|_\alpha = 0$
and that each letter occurs at most {once} in each of the words $w_1$,
$w_2$ and $w_3$. \qed
\end{proof2}
\end{document}